\documentclass[prd,aps,twocolumn,preprintnumbers, showpacs, nofootinbib,superscriptaddress,notitlepage]{revtex4-1}
\usepackage{amsmath,graphicx,epsfig,amssymb}
\usepackage{inputenc}
\usepackage[usenames]{color}
\usepackage{ulem} 
\usepackage{bigstrut}
\usepackage{slashed}
\usepackage{multirow}
\usepackage{subfigure}
\usepackage{dsfont}
\usepackage{xcolor}
\usepackage{hyperref}
\usepackage{comment}
\usepackage{amsthm}
\usepackage{cleveref}
\newtheorem{theorem}{Theorem}[section]

\allowdisplaybreaks

\newcommand{\norm}[1]{\left\lVert #1 \right\rVert}

\begin{document}

\title{Ill-Posedness in Limited Discrete Fourier Inversion and Regularization for Quasi Distributions in LaMET}

\author{Ao-Sheng Xiong}
\affiliation{Frontiers Science Center for Rare Isotopes, and School of Nuclear Science and Technology, Lanzhou University, Lanzhou 730000, China}

\author{Jun Hua}
\email{Corresponding author: junhua@scnu.edu.cn}
\affiliation{Key Laboratory of Atomic and Subatomic Structure and Quantum Control (MOE), 
Guangdong Basic Research Center of Excellence for Structure and Fundamental Interactions of Matter, 
Institute of Quantum Matter, South China Normal University, Guangzhou 510006, China}
\affiliation{Guangdong-Hong Kong Joint Laboratory of Quantum Matter, 
Guangdong Provincial Key Laboratory of Nuclear Science, Southern Nuclear Science Computing Center, 
South China Normal University, Guangzhou 510006, China}

\author{Ting Wei} 
\affiliation{School of Mathematics and Statistics, Lanzhou University, Lanzhou 730000, China}

\author{Fu-Sheng Yu}
\affiliation{Frontiers Science Center for Rare Isotopes, and School of Nuclear Science and Technology, Lanzhou University, Lanzhou 730000, China}
\author{Qi-An Zhang}
\email{Corresponding author: zhangqa@buaa.edu.cn}
\affiliation{School of Physics, Beihang University, Beijing 102206, China}
\author{Yong Zheng}
\affiliation{Frontiers Science Center for Rare Isotopes, and School of Nuclear Science and Technology, Lanzhou University, Lanzhou 730000, China}

\begin{abstract}
We systematically investigated the limited inverse discrete Fourier transform of the quasi distributions from the perspective of inverse problem theory. This transformation satisfies two of Hadamard's well-posedness criteria, existence and uniqueness of solutions, but critically violates the stability requirement, exhibiting exponential sensitivity to input perturbations. To address this instability, we implemented Tikhonov regularization with L-curve optimized parameters, demonstrating its validity for controlled toy model studies and real lattice QCD results of quasi distribution amplitudes. 
The reconstructed solutions is consistent with the physics-driven $\lambda$-extrapolation method. 
Our analysis demonstrates that the inverse Fourier problem within the large-momentum effective theory (LaMET) framework belongs to a class of moderately tractable ill-posed problems, characterized by distinct spectral properties that differ from those of more severely unstable inverse problems encountered in other lattice QCD applications.
Tikhonov regularization establishes a rigorous mathematical framework for addressing the underlying instability, enabling first-principles uncertainty quantification without relying on ansatz-based assumptions.
\end{abstract}

\maketitle


\section{Introduction}\label{sec:introduction}

In the past few decades, lattice QCD has provided the most systematic first-principles framework for investigating hadronic structures, including parton distribution functions (PDFs), light-cone distribution amplitudes (LCDAs), generalized parton distributions (GPDs), and transverse momentum distributions (TMDs), etc. However, constrained by its formulation in Euclidean space with imaginary time, the standard lattice QCD approach cannot be applied for directly calculating the dynamical correlation functions of quark and gluon fields with real-time dependence defined on the light-front (LF). 

A traditional approach to circumvent this limitation involves the operator product expansion (OPE), through which a few low-order moments can be calculated \cite{Martinelli:1987zd,Detmold:2001dv,LHPC:2002xzk,Detmold:2005gg,Braun:2006dg,Braun:2015axa,Chambers:2017dov,RQCD:2019osh,Shindler:2023xpd}. 
In recent years, LaMET \cite{Ji:2013dva,Ji:2020ect} has introduced an alternative choice, enabling direct reconstruction of the full momentum fraction dependence for LF distributions of hadrons via quasi-distributions combined with systematic matching procedures derived from the effective theory framework \cite{Xiong:2013bka,Lin:2014zya,Ji:2014hxa,Ji:2014lra,Sufian:2014jma,Ji:2015qla,Xiong:2015nua,Alexandrou:2016eyt,Chen:2016utp,Lin:2016qia,Yang:2016nfc,Yang:2016plb,Alexandrou:2017dzj,Alexandrou:2017huk,Alexandrou:2017qpu,Chen:2017mie,Chen:2017mzz,Constantinou:2017sej,Ishikawa:2017iym,Ji:2017oey,Ji:2017rah,Lin:2017ani,Stewart:2017tvs,Xiong:2017jtn,Zhang:2017bzy,Zhang:2017zfe,Alexandrou:2018eet,Alexandrou:2018pbm,Alexandrou:2018yuy,Chen:2018xof,Ebert:2018gzl,Fan:2018dxu,Izubuchi:2018srq,Ji:2018hvs,LatticeParton:2018gjr,Xu:2018mpf,Lin:2018pvv,Liu:2018hxv,Liu:2018tox,Zhang:2018diq,Zhang:2018nsy,Zhang:2018rls,Zhao:2018fyu,Alexandrou:2019dax,Alexandrou:2019lfo,Chai:2019rer,Chen:2019lcm,Constantinou:2019vyb,Ebert:2019tvc,Liu:2019urm,Shanahan:2019zcq,Wang:2019msf,Zhang:2019qiq,Alexandrou:2020qtt,Alexandrou:2020uyt,Alexandrou:2020zbe,Bhattacharya:2020cen,Bhattacharya:2020jfj,Chai:2020nxw,Chen:2020arf,Chen:2020iqi,Chen:2020ody,Ebert:2020gxr,Fan:2020nzz,Gao:2020ito,Hua:2020gnw,Ji:2020brr,Ji:2020ect,Ji:2020jeb,LatticeParton:2020uhz,Lin:2020fsj,Lin:2020rxa,Lin:2020ssv,Shanahan:2020zxr,Shugert:2020tgq,Vladimirov:2020ofp,Zhang:2020dkn,Zhang:2020gaj,Zhang:2020rsx,Alexandrou:2021bbo,Alexandrou:2021oih,Bhattacharya:2021moj,Bhattacharya:2021rua,Constantinou:2021nbn,Dodson:2021rdq,Gao:2021hxl,Gao:2021dbh,LatticePartonLPC:2021gpi,Li:2021wvl,Lin:2021brq,Lin:2021ukf,Scapellato:2021uke,Schlemmer:2021aij,Shanahan:2021tst,Bhattacharya:2022aob,Constantinou:2022fqt,Deng:2022gzi,Ebert:2022fmh,Gao:2022iex,Gao:2022uhg,LatticeParton:2022xsd,LatticeParton:2022zqc,LatticePartonCollaborationLPC:2022myp,LatticePartonLPC:2022eev,Scapellato:2022mai,Schindler:2022eva,Zhang:2022xuw,Alexandrou:2023ucc,Avkhadiev:2023poz,Bhattacharya:2023jsc,Bhattacharya:2023nmv,Bhattacharya:2023tik,Cichy:2023dgk,Deng:2023csv,Gao:2023ktu,Gao:2023lny,Holligan:2023jqh,Holligan:2023rex,Ji:2023pba,LatticeParton:2023xdl,LatticePartonLPC:2023pdv,Lin:2023gxz,Zhao:2023ptv,Liu:2023onm,Avkhadiev:2024mgd,Baker:2024zcd,Bollweg:2024zet,Chen:2024rgi,Cloet:2024vbv,Ding:2024saz,Gao:2024fbh,Good:2024iur,Han:2024min,Holligan:2024umc,Holligan:2024wpv,Ji:2024hit,LatticeParton:2024mxp,Wang:2024wwa,LatticeParton:2024vck,LatticeParton:2024zko,Miller:2024yfw,Mukherjee:2024xie,Spanoudes:2024kpb,Zhang:2024omt,Bollweg:2025ecn,Bollweg:2025iol,Wang:2025uap,Han:2025odf,Holligan:2025ydm,LPC:2025spt,Zhang:2025hvf,Ji:2025mvk,Chen:2025cxr}.
Although OPE achieves remarkable precision for low-order moments, the scarcity of accessible moments (typically up to the fourth moment in practical calculations) fundamentally restricts the resolution of complete momentum fraction dependence due to the inherent inverse problem. Conversely, LaMET achieves continuum momentum fraction resolution based on three critical advancements: 1) a reliable nonperturbative renormalization of the quasi distributions, 2) a physics-driven systematic extrapolation for the long-range correlations, and 3) an efficient large-momentum expansion of the light-cone observables.

Inverse problems constitute a foundational challenge in lattice QCD, stemming primarily from their inherently ill-posed mathematical nature. The typical manifestation of this challenge emerges in the analytic continuation from Euclidean to real-time observables, exemplified by spectral function reconstruction from Euclidean correlators \cite{Li:2020xrz,Wang:2021jou}. This procedure requires solving the ill-conditioned inverse problem: 
\begin{align}
    G_E(\tau) = \int_0^\infty \frac{d\omega}{\pi} K(\tau,\omega) \rho(\omega) 
\end{align}
where $K(\tau, \omega)=\frac{\omega}{\omega^2+\tau^2}$, and $\rho(\omega)$ denotes the sought-after spectral function. Three compounding factors amplify the ill-posedness: (1) exponential damping of high-frequency information in Euclidean time $(\tau)$, (2) discrete and noise-contaminated lattice data for $G_E(\tau)$, and (3) complex and non-analytic structures (e.g., narrow resonances, threshold discontinuities) in $\rho(\omega)$. Analogous inverse problems plague QCD sum rules when implementing dispersion relations between OPE and hadronic spectral densities \cite{Zhao:2024drr,Li:2020ejs}.

To address these challenges, regularization methods utilize the Maximum Entropy Method (MEM) or Bayesian inference frameworks incorporating prior distributions to stabilize the inversion process through systematic constraint implementation, thereby enabling physically viable spectral reconstructions with quantified uncertainties. However, these inverse problems remain particularly intractable due to the coexistence of narrow resonance structures and continuum contributions in QCD spectral functions. For example, in \cite{Asakawa:2000tr} the MEM approach has produced a very precise prediction for the peak position of the ground state, while it exhibits exponentially growing errors for continuum structures.

The inverse problem in reconstructing partonic distributions from moments has persisted as a fundamental puzzle \cite{Li:2022qul}. Recently, concerns have also been raised that LaMET methods may inherit analogous challenges \cite{Dutrieux:2025jed}, rooted in the mathematical structure of the finite-range Fourier transformation. However, the aforementioned work overlooked the non-static nature of ill-posedness inherent in this inverse problem. The instability of Fourier transforms is intrinsically constrained by both data convergence and precision. From a physical perspective, imposing QCD-motivated constraints on long-range coordinate-space correlations, while allowing for systematic extrapolation in this regime, can enhance data convergence, thereby mitigating ill-posedness \cite{Chen:2025cxr}. From a mathematical standpoint, the critical questions are: (1) whether the inverse problem is genuinely ill-posed, and more importantly, (2) whether its degree of ill-posedness reaches a severity that renders mathematical treatment impractical. This underscores the necessity to recognize that not all inverse problems exhibit equivalent levels of pathological instability.

In this paper, we systematically investigate the mathematical framework of ill-posed inverse problems, beginning with their rigorous formulation. More precisely, we analyze the limited discrete Fourier transformation inherent in quasi distributions as a concrete manifestation of such inverse problems. To resolve this ill-posedness, we implement the Tikhonov regularization method, a gold-standard approach in inverse problem theory, with particular emphasis on quantifying solution stability. Our demonstration employs both synthetic toy models (with controlled error structures) and real lattice QCD data for the pion quasi distribution amplitude (quasi DA), providing complementary validation perspectives. The paper is organized as follows: Sec. \ref{sec:Mathematical framework} establishes the mathematical foundations of inverse problems and introduces the Tikhonov regularization paradigm for their systematic treatment. Sec. \ref{sec:Toy-model} conducts rigorous testing of the regularization method for toy models featuring both correlated and uncorrelated noise profiles. Sec. \ref{sec:Pion-quasi-DA} presents the  application of this framework to state-of-the-art lattice data for pion quasi DAs, revealing practical implementation challenges and solutions. We conclude in Sec. \ref{sec:Summary-outlook} and discuss promising inverse problem techniques for partonic structure calculations. The detailed proof of uniqueness is provided in Appendix \ref{sec:Proof of the uniqueness}.


\section{Mathematical framework}\label{sec:Mathematical framework}
\subsection{The definition of inverse problems}

To establish rigorous understanding of the inverse problem, we first contextualize their fundamental duality with forward problems. Forward problems involve computing the output from specified inputs and known physical or mathematical models, essentially determining effects from known causes. For example, a forward problem systematically computes observable outputs $y \in Y$ of deterministic operator $\mathcal{F}: \Theta \to Y$ acting on input parameters $\theta \in \Theta$, mathematically expressed as
\begin{align}
    \mathcal{F}(\theta) = y.
\end{align}
This represents causal prediction from complete physical characterization.  Conversely, inverse problems attempt to reconstruct unknown parameters $\theta$ from indirect observations $y^{\delta}$ through operator inversion $\mathcal{F}^{-1}$, where superscript $\delta$ denotes data perturbation satisfying $\|y - y^{\delta}\|_Y \leq \delta$. Such inverse formulations pervade scientific domains \cite{Application in sci and eng,inverse problem application,nonlinear problem application}, particularly when direct parameter measurement proves technologically constrained or economically impractical. A typical scenario arises when certain parameters of a system cannot be directly measured or when the costs of measurement become prohibitive. In such cases, inferring system parameters from measurable experimental data using inversion techniques becomes a possible option.

The critical step to identify an inverse problem is to check Hadamard's well-posedness criteria \cite{Kirsch-2011}: 1) {\bf Existence}: for every $y\in Y$ there is (at least one) $\theta\in \Theta$ such that $\mathcal{F}(\theta) = y$; 2) {\bf Uniqueness}: for every $y\in Y$ there is at most one $\theta\in \Theta$ with $\mathcal{F}(\theta) = y$; and 3) {\bf Stability}: the solution $\theta$ depends continuously on $y$; that is, for every sequence $(\theta_n) \subset \Theta$ with $\mathcal{K}\theta_n \to \mathcal{K}\theta (n\to \infty)$, it follows that $\theta_n\to \theta (n\to \infty)$. Problems violating any criterion are termed {\bf ill-posed} \cite{Engl-1996}. Inverse problems intrinsically manifest ill-posedness through: solution non-uniqueness arising from underdetermined measurements, or nonexistence of a specific solution due to overdetermined noisy data, and instability via discontinuous dependence on $y^{\delta}$ where measurement errors $\delta$ significantly amplify in parameter space.

Fundamental computational challenges emerge when applying classical methods to ill-posed problems. First, the operator $\mathcal{F}$ mapping the data space $\Theta$ to the parameter space $Y$ is not a bijection, leading to the failure of existence or uniqueness. Second, the operator $\mathcal{F}$ lacks a bounded inverse operator, such that even small perturbations $\delta$ result in significant variations in the solution. These limitations stem from inherent operator ill-posedness rather than theoretical unsolvability. They can be effectively addressed by applying mathematical regularization techniques, which aim to resolve the ill-posedness and yield a unique and stable approximate solution \cite{linear inverse problem with discrete data I,linear inverse problem with discrete data II}. 

Methods for obtaining unique stable approximate solutions to ill-posed inverse problems are collectively known as regularization methods. Over the past few decades, a wide variety of regularization methods have been developed for inverse problems \cite{Regularization-tools,Modern regularization methods}. Typically, the choice of method depends on the specific nature of the inverse problem being addressed. In this paper, we focus on the widely adopted Tikhonov regularization method which serves as an exemplary framework to clearly illustrate both the ill-posed nature of inverse problems and the fundamental effectiveness of regularization techniques. Furthermore, it has been rigorously proven that as the error in the input data approaches zero, the solution obtained via Tikhonov regularization converges to the true solution \cite{Tik-L-curve}. These mathematical properties make solving ill-posed problems no longer insurmountable. Specifically, regularization theory resolves this through controlled problem modification, constructing solution families $\{\theta_{\alpha}^{\delta}\}$ parameterized by $\alpha > 0$ that satisfy:
\begin{align}
    \lim_{\substack{\delta \to 0 \\ \alpha \to 0}} \|\theta_{\alpha}^{\delta} - \theta_{\mathrm{true}}\|_\Theta = 0
\end{align}
and Tikhonov regularization  exemplifies this approach through constrained optimization
\begin{align}
    \theta_{\alpha}^{\delta} = \arg\min_{\theta} \|\mathcal{F}(\theta) - y^{\delta}\|^2_Y + \alpha\|\Gamma\theta\|^2_\Theta,
\end{align}
where $\Gamma$ encodes prior solution constraints. Its effectiveness for limited Fourier inversion stems from three mechanisms: spectral control via parameter $\alpha$, natural compatibility with the priors, and error-adaptive convergence  where solution accuracy improves with reduced data error $\delta$ and increased solution regularity (Sobolev differentiability \cite{Computational methods for IP}).

The efficiency of regularization methods depends on both the priori smoothness characteristics of the true solution (quantifiable through H\"older continuity or Sobolev regularity metrics) and the posteriori error magnitude in the measured data. In general, the smaller the input errors, and the smoother the true solution, the easier it is to solve the ill-posed problem \cite{Kirsch-2011,Engl-1996}. Additionally, the reconstruction fidelity improves with both decreasing data uncertainties and increasing solution regularity. Fortunately, our analysis benefits from two inherent advantages of hadrons partonic distributions (such as PDFs or LCDAs): their positivity and analyticity properties in momentum space, combined with the convergence of corresponding coordinate-space correlation functions. The latter exhibit exponential decay with increasing correlation length, ensuring robust convergence within finite spatial domains \cite{Chen:2025cxr}. These characteristics enable significantly enhanced tractability of the limited discrete Fourier transform problem for quasi distributions when input errors remain well-controlled.

\subsection{The inversion for limited Fourier transform}

Within the framework of short-distance factorization \cite{Radyushkin:2017cyf,Orginos:2017kos,Dutrieux:2025jed}, the inverse problem is universally recognized as both inherent and critical due to the limitations of reliable data obtained from lattice QCD calculations in coordinate space, specifically, constraints in both the finite precision of the data and the restricted correlation lengths \cite{Dutrieux:2025jed,Braun:2007wv,Radyushkin:2017cyf,Ma:2017pxb}. The ill-posed nature of this inverse problem fundamentally originates from the inversion of limited discrete Fourier transformation.

The limited inverse Fourier transform takes the form
\begin{align}
    g(\lambda) = \int dx e^{i x \lambda} f(x), \label{eq:limit_Fourier_Transform}
\end{align}
where $\lambda$ denotes the Lorentz-invariant parameter in coordinate space (also referred to  as the "Ioffe time" \cite{Orginos:2017kos,Radyushkin:2017cyf}), $g(\lambda)$ denotes the renormalized matrix element of a nonlocal Euclidean equal-time correlation operator, while $f(x)$ represents the corresponding momentum-space parton distribution with momentum fraction $x$. When the complete distribution of $g(\lambda)$ is precisely known, one can easily obtain the result of $f(x)$ through Fourier inversion. However, due to the limitations of lattice QCD calculations, one can only obtain results for $g(\lambda)$ at discrete and finite values of $\lambda$: $\lambda\in\{\lambda_0, \lambda_1, \cdots, \lambda_n\}$ with $\lambda_k\leq\lambda_n, ~k=0,1, \cdots, n$. Mathematically, Eq.~(\ref{eq:limit_Fourier_Transform}) represents a Fredholm integral equation of the first kind $g=\mathcal{K}f$, where the integral operator $\mathcal{K}$ indicates that the integral kernel $e^{i x \lambda}$ acts on the solution $f(x)$. The inverse problem of recovering $f(x)$ from noisy integral measurements $g^{\delta}(\lambda)$ (where $ \| g-g^{\delta} \|_{l^2} \leq\delta$ and $l^2$ denotes the space of square-summable sequences) represents a canonical example of ill-posed operator inversion \cite{Xiong:2022uwj,Asakawa:2000tr}.

More specifically, the discretization of lattice QCD simulations necessitates employing discrete Fourier transforms (DFT) formally expressed as followed:
\begin{align}{\label{eq:discrete_Fourier_Transform}}
   g(\lambda) =  K f(x),    
\end{align}
with DFT kernel matrix $ K$, which is explicitly constructed as:
\begin{align}{\label{eq:Complex_matrix_K}}
     K = \delta x \cdot \exp\left(i \mathbf{X}{n_x} \otimes \mathbf{\Lambda}{n_\lambda}^\top\right),
\end{align}
with discrete parameters defined by:
\begin{align}
\mathbf{X}{n_x} &= \left\{x_{\mathrm{min}},\ x_{\mathrm{min}}+\delta x,\ \ldots,\ x_{\mathrm{max}} \right\}, \nonumber\\
\mathbf{\Lambda}{n\lambda} &= \left\{ -\lambda_{\mathrm{max}},\ -\lambda_{\mathrm{max}}+\delta\lambda,\ \ldots,\ \lambda_{\mathrm{max}} \right\}, 
\end{align}
where $\delta x$ and $\delta\lambda$ represent the resolutions in momentum and coordinate space.

A fundamental prerequisite for solving inverse problems lies in establishing their well-posedness according to Hadamard's criteria. For the limited Fourier transform problem in Eq.~(\ref{eq:limit_Fourier_Transform}), we first analyze its existence and uniqueness. The Paley-Wiener theorem \cite{W-P theorem} establishes existence under specific analyticity constraints: if $f(x)$ possesses compact support in $x\in[x_{\mathrm{min}},x_{\mathrm{max}}]$ and its Fourier transform is an entire function on the complex plane that decays exponentially, then the solution $f(x)$ satisfies existence. On the other hand, the partonic distributions of hadrons exhibit precisely such required behavior, demonstrating exponential decay at large spatial separations which guarantees that exponentially decaying correlators correspond to analytic momentum distributions with finite support\cite{Ji:2020brr,LatticePartonLPC:2021gpi,Gao:2021dbh,Zhang:2023bxs}
\begin{align}
    g(\lambda)\propto e^{-m_{\mathrm{eff}}|z|},\quad |z|\to\infty 
\end{align}
where $m_{\mathrm{eff}}$ characterizes an effective mass governing the exponential decay behavior at large $\lambda$. The rigorous proof of the uniqueness theorem \ref{thm:The Proof of the uniqueness} is provided in the Appendix. The solution space is defined as the space of square-integrable functions, which is consistent with the behavioral properties of quasi-DA.  Although the theoretical existence and uniqueness are guaranteed under ideal conditions, practical implementations face significant challenges due to both statistical and systematic errors. The statistical uncertainty $\delta g$, which scales as $\mathcal{O}(1/\sqrt{N_{\mathrm{cfg}}})$ with the number of lattice configurations, combines with systematic errors arising from the truncation at $\lambda<\infty$ to critically impact numerical convergence.

To quantify the numerical instability inherent in discrete Fourier inversion, we analyze the real component of the matrix $K_{\mathrm{re}}$ through the standard singular value decomposition (SVD) without any truncation. The SVD formalism decomposes the matrix as follows:
\begin{align}\label{SVD of matrix K}
K_{\text{re}} = U \Sigma V^T = \sum_{i=1}^{n} \sigma_i u_i  v^T_i, \quad n = \min(n_\lambda, n_x),
\end{align}
where $U = [u_1, \ldots, u_n] \in \mathbb{R}^{n_\lambda \times n}$ and $V=[v_1, \ldots, v_n] \in \mathbb{R}^{n_x \times n}$ denote singular vectors with orthonormal columns, and satisfies $U^T U = V^T V = I_n$. $\Sigma=\text{diag}(\sigma_1, \ldots, \sigma_n)$ denotes the singular values of $K_{\mathrm{re}}$ ordered as $\sigma_1 \geq \cdots \geq \sigma_n \geq 0$ \cite{Singular values and condition numbers}. The formal solution to $g=K_{\mathrm{re}}f$ becomes
\begin{align}\label{eq:SVD solution}
f = \sum_{i=1}^n \frac{u_i^T g}{\sigma_i} v_i.
\end{align}
The condition number $\kappa(K_{\text{re}})$, defined as $\sigma_1/\sigma_n$, quantifies the numerical stability of the system. When the largest singular value $\sigma_1$ and the smallest $\sigma_n$ differ by orders of magnitude, the resulting large condition number induces severe numerical instability in the reconstructed solution $f(x)$, as high-frequency noise is substantially amplified during the solution process.

\begin{figure}
\centering
\includegraphics[width=1.1\linewidth]{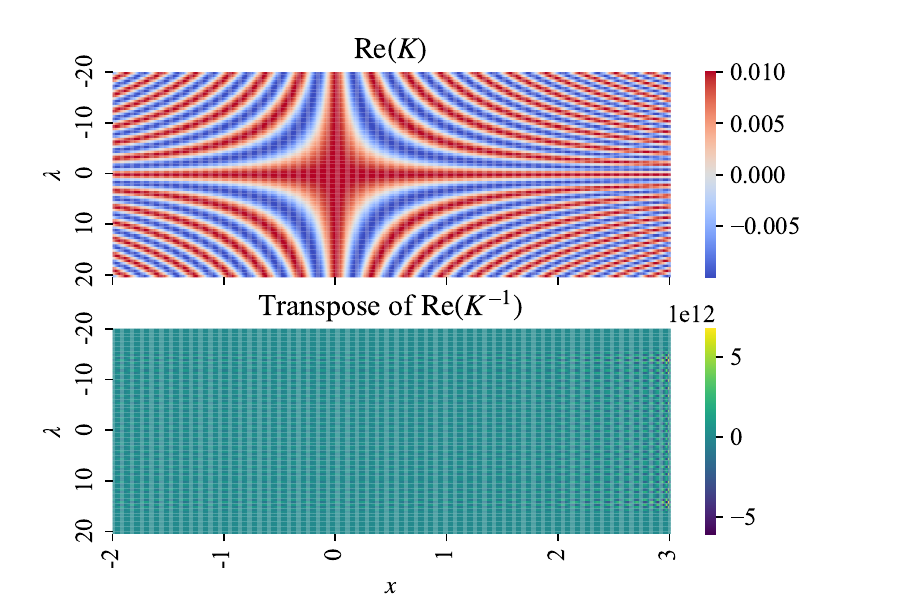}
\caption{Illustration of the matrix $K_{\mathrm{re}}$ and its generalized inverse. For clarity, the lower panel displays the transpose of $K_{\mathrm{re}}^{-1}$. The parameters correspond to those specified in Eq.~(\ref{eq:inputs_for_DFT}). It should be noted that the amplitude in the upper panel is on the order of 0.01, whereas after inversion, the magnitude in the lower panel increases to the order of $10^{12}$. This amplification originates from the singularity inherent in the inverse DFT.}
 \label{fig:matrix-K-and-its-inverse}
\end{figure}

To elucidate this instability mechanism, we present a numerical example with the following parameters:
\begin{align}\label{eq:inputs_for_DFT}
&x_{\mathrm{min}}=-2,~x_{\mathrm{max}}=3,~\delta x=0.01,\\
        & \lambda_{\mathrm{min}}=-20,~\lambda_{\mathrm{max}}=20,~\delta \lambda=0.5,
\end{align}
where the matrix $K_{\mathrm{re}}$ and its generalized inverse $K_{\mathrm{re}}^{-1}$ exhibit the spectral structure shown in Fig.~\ref{fig:matrix-K-and-its-inverse}. The singular spectrum of $K_{\mathrm{re}}$ decays as:
\begin{align}
    & \sigma_1=0.3545 \geq \cdots \geq \sigma_n=3.8\times 10^{-18},
\end{align}
yielding a huge condition number $\kappa(K_{\text{re}})=\sigma_1/\sigma_n\sim10^{17}$. This ill-conditioning arises from near-linear dependencies between matrix rows/columns, which is regarded as a hallmark of discrete Fourier systems. A larger condition number indicates a more unstable matrix $K_{\mathrm{re}}$ \cite{Engl-1996, Regularization-tools}, as can be easily observed from its SVD representation  Eq.~(\ref{eq:SVD solution}). Furthermore, a larger condition number also corresponds to a faster decay of the singular values, both of which are indicative of instability in the solution. So the inverse DFT is unstable. It is worth noting that the singular values decay more slowly than exponentially. Mathematically, this type of instability is considered mild and can be effectively handled by regularization methods \cite{Kirsch-2011,Engl-1996}. Using SVD analysis and Picard's criterion \cite{Tik for Fredholm equation}, we quantify the instability:
\begin{align}{\label{eq:diff-before-regularization}}
     \norm{f^{\delta} - f_t}_{l^2}^2 \equiv  \sum_{i=1}^{n} (\frac{u_i^T   (g^\delta-g)}{\sigma_i}v_i)^2 = \sum_{i=1}^{n} (\frac{u_i^T \delta}{\sigma_i} v_i)^2 \to \infty,
\end{align}
where $g$ denote the exact error-free input data, and $f_t$ denotes its true solution. $g^\delta=g+\delta$ represents noise-corrupted data with $\delta \sim \mathcal{N}(0,\epsilon^2)$, and $f^\delta$ is the solution corresponding to the noisy data. 
The continuum limit of vanishing dispersions (\(\delta x, \delta\lambda \to 0\)) exposes intrinsic mathematical pathologies of the non-regularized inversion. As the discretization is refined, inverse of the singular values $\sigma_i$ of the transformation matrix become unbounded ($\sigma_i^{-1} \to \infty$), inducing a divergence between solutions reconstructed from noise-contaminated data $g^\delta$ and the true solution $f_t$.
Consequently, traditional solution techniques (such as SVD without regularization) are inadequate, as arbitrarily small measurement errors $\|\delta\|_{l^2}$ become extremely amplified through high-frequency singular modes.

\subsection{Tikhonov regularization method}
To address such ill-posed inverse problems, we apply the Tikhonov regularization method~\cite{Xiong:2022uwj}. The basic idea is to find the minimizer $f_\alpha^\delta$, known as the regularized solution, by minimizing the following Tikhonov functional
\begin{align}{\label{Tikhonov Fountional}}
   J(f) = \norm{Kf -g^{\delta}}_{l^2}^2 + \alpha \norm{f}^2_{l^2},
\end{align}
where 
$\alpha >0$ denotes the regularization parameter. By taking the variation of the above functional, one find that the regularized solution satisfies the following equation
\begin{align}\label{eq: (KK+aI)f=Kg}
     (K^{\dagger} K + \alpha I)f^{\delta}_{\alpha} = K^{\dagger} g^{\delta},
\end{align}
where $K^{\dagger}$ is the transpose of matrix $K$ and $I$ is a diagonal identity matrix \cite{Tik for Fredholm equation,Tik-L-curve}. From the Eq.~(\ref{eq: (KK+aI)f=Kg}), it is clear that the Tikhonov regularization method works by introducing a small regularization parameter $\alpha$ to mitigate the effects of ill-posedness caused by a large condition number, i.e., 
\begin{align}
 \frac{\sigma_1}{\sigma_n} \Rightarrow    \frac{\sigma_1 + \alpha }{ \sigma_n + \alpha} \approx \frac{\sigma_1}{\alpha},
\end{align}
This basic form and mechanism are quite similar to the use of regulators in perturbative QCD.

Moreover, as seen in the Tikhonov functional, the regularization parameter $\alpha$ significantly influences the outcome. If $\alpha$ is too large, the solution may underfit the data, leading to poor convergence. Conversely, if $\alpha$ is too small, the solution becomes unstable, resulting in overfitting. An appropriate choice of $\alpha$ should minimize the discrepancy between the regularized solution $f^{\delta}_{\alpha}$ and the true solution $f_t$, given a fixed level of input error. Therefore, a careful and rigorous determination of $\alpha$ is essential \cite{Kirsch-2011, Engl-1996}.

Mathematically, one can choose the regularization parameter by incorporating prior knowledge of the solution, which offers guidance on the convergence behavior of the regularization method.
If the prior information of the true solution $f_t$ is known, i.e., $f_t=K^{\dagger} v, \|v\|_{l^2} \leq E$,  we have,
\begin{align}
     \norm{f_{\alpha}^{\delta} -f_t}_{l^2} \leq \frac{\delta}{2 \sqrt{\alpha}} + \frac{\sqrt{\alpha} E}{2},
\end{align}
where $f_t$ is the true solution, $\delta$ is the error level and $E$ is the upper bound of the true solution. 
Then, $\alpha$ can be chosen as:
\begin{align}
     \alpha = \delta/E, 
\end{align}
such that one obtains
\begin{align}{\label{eq:diff-after-regularization}}
    \norm{f_{\alpha}^{\delta} -f_t}_{l^2} \leq \sqrt{\delta E} \overset{\delta \to 0}{\longrightarrow} 0,
\end{align}
which implies that 
the regularized solution converges to the true solution as the error approaches zero \cite{Tik for Fredholm equation,Tik-L-curve}. By comparing Eq.~(\ref{eq:diff-before-regularization}) and Eq.~(\ref{eq:diff-after-regularization}), it is evident that the  
Tikhonov regularization method effectively mitigates ill-posedness. Therefore, one can prove that the inverse problem is solvable.

In practice, determining the values of $\delta$ and $E$ can be challenging. However, alternative effective methods are available for selecting the regularization parameter \cite{The posterier parameter choice}. The L-curve method is a valid approach for this purpose \cite{HPC-1993}. 
The core idea involves comparing the curves of $\| K f^\delta_\alpha - g^\delta \|^2_{l^2}$ and $\|f^\delta_\alpha\|^2_{l^2}$ on a Log-Log scale to identify the regularization parameter. This method aims to balance the relative weights of the two components in the Tikhonov functional, ensuring that neither term dominates the minimization process. It is named after the characteristic L-shaped curve that appears when plotted on the scale. The L-profile is generated by the compensatory variation, where the term $\| K f^\delta_\alpha - g^\delta 
\|^2_{l^2}$ decreases monotonically as $\alpha$ decreases, while the term $\|f^\delta_\alpha\|^2_{l^2}$ exhibits an opposite trend \cite{L curve for discrete ill-posed problem}. This method selects the regularization parameter $\alpha$ by maximizing the curvature of the L-curve. Let $\rho=\log(\| K f^\delta_\alpha - g^\delta 
\|^2_{l^2})$ and $\theta=\log(\|f^\delta_\alpha\|^2_{l^2})$, with the curvature defined as function of the parameter $\alpha$:
\begin{align}
    & L(\alpha) = \frac{\rho' \theta''-\rho'' \theta'}{((\rho')^2+(\theta')^2)^{\frac{3}{2}}}, \nonumber\\
    & \alpha^*=\arg \left \{\sup _{\alpha>0} L(\alpha) \right \},
\end{align}
where $'$ denotes the partial derivative with respect to $\alpha$, and  $\alpha^*$ represents the chosen regularization parameter.

Therefore, our analysis establishes that the DFT inverse problem is ill-posed: while satisfying the existence and uniqueness criterion, it fundamentally violates stability due to tremendously diverging error amplification (\(\| \delta f \|_{l^2} \propto \kappa(K_{\text{re}}) \| \delta g \|_{l^2} \)). Crucially, we find that Tikhonov regularization restores well-posedness by introducing a regularization parameter, yielding solutions that 
converge to the true distribution as measurement errors vanish (\(\| \delta g \|_{l^2} \to 0\)).  

\section{Toy Models}\label{sec:Toy-model}

\begin{figure}
\centering
\includegraphics[width=0.95\linewidth]{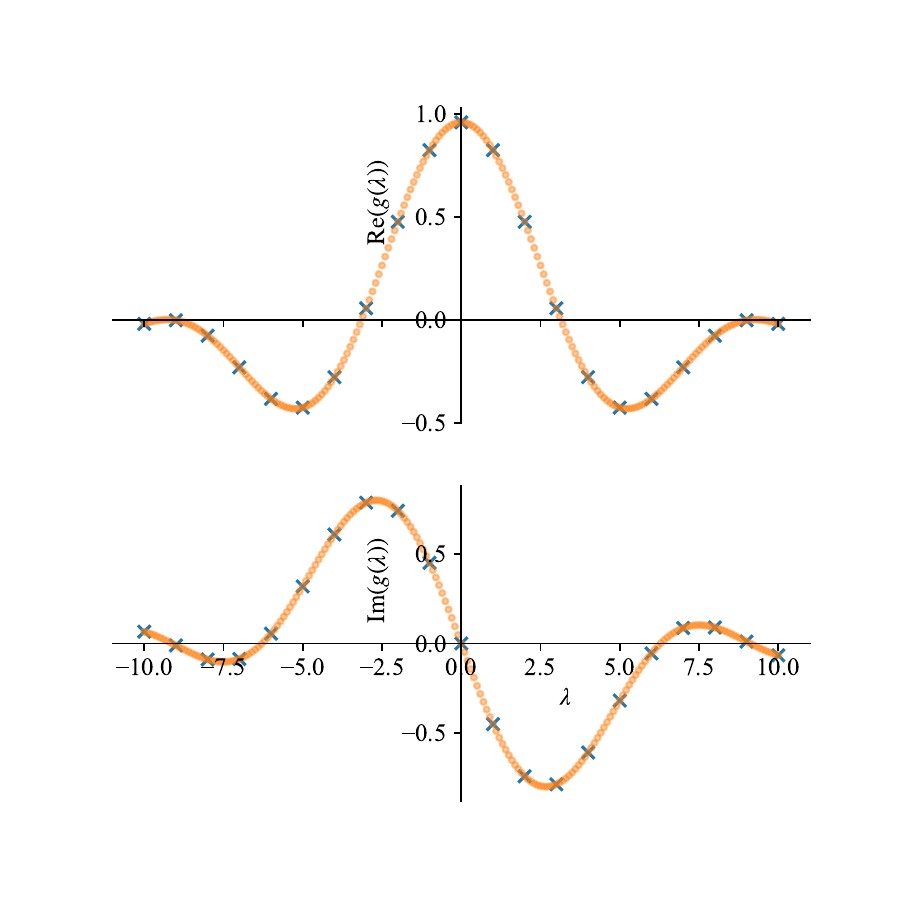}
\caption{Examples of discretization for $g(\lambda)$ used as inputs to the DFT. The figure shows cases with $\lambda_{\mathrm{max}}=10$ and $\delta\lambda=1$ or $0.1$, represented by blue crosses and orange circles respectively.}
 \label{fig:toy_gl}
\end{figure}

To systematically investigate the ill-posed nature of limited discrete Fourier inversion and validate the regularization framework, we construct a controlled toy model with a known ground truth. Adopting the asymptotic LCDA form as motivation, we define the true momentum-space solution as:
\begin{align}
    f_t(x) = 6x(1-x), \quad x\in[0,1],
\end{align}
which also represent the tree-level approximation of the quasi distribution amplitude for light mesons. The synthetic coordinate space distribution $g(\lambda)$ is generated through the DFT, shown as Fig.~\ref{fig:toy_gl}, to emulate the idealized lattice QCD measurements. The result for the inverse DFT $f_r(x)$, reconstructed through the Moore-Penrose pseudoinverse $K^{-1}$, is independent of the discretization parameters. As illustrated in Fig.~\ref{fig:DFT_params}, we systematically compare the effects of varying the parameters $\delta\lambda$, $\lambda_{\text{max}}$, and $\delta x$ on the reconstructed $f_r(x)$, employing a controlled variable approach. 
The comparison demonstrates that all parameter selections produce numerically indistinguishable solutions $f_r(x)$, which accurately reproduce the target distribution $f_t(x)$.
This suggests that discretization and truncation artifacts do not inherently lead to ill-posedness in statistical-noise-free systems. The observed stability stems from the analyticity of $g(\lambda)$, which satisfies the spectral consistency conditions of the Sampling theorem \cite{Sampling theorem}. In other words, the inversion for the limited Fourier transform satisfies the existence and uniqueness conditions of the criteria for well-posedness.

\begin{figure}
\centering
\includegraphics[width=0.75\linewidth]{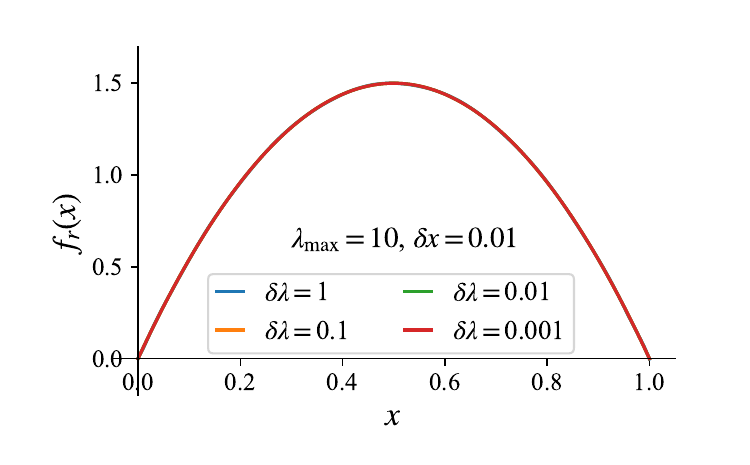}
\includegraphics[width=0.75\linewidth]{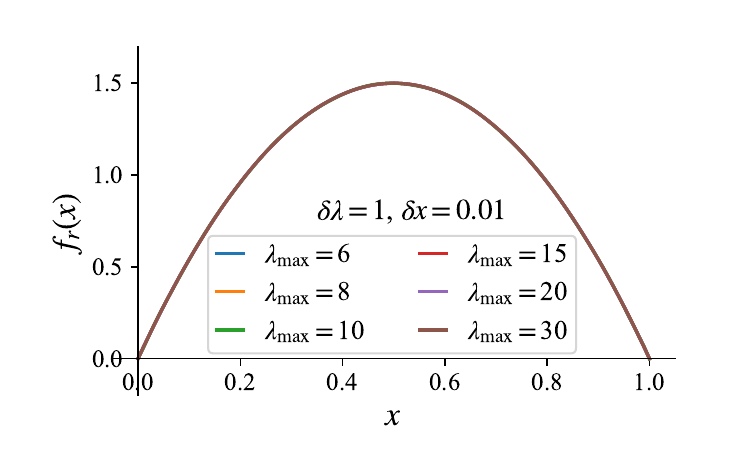}
\includegraphics[width=0.75\linewidth]{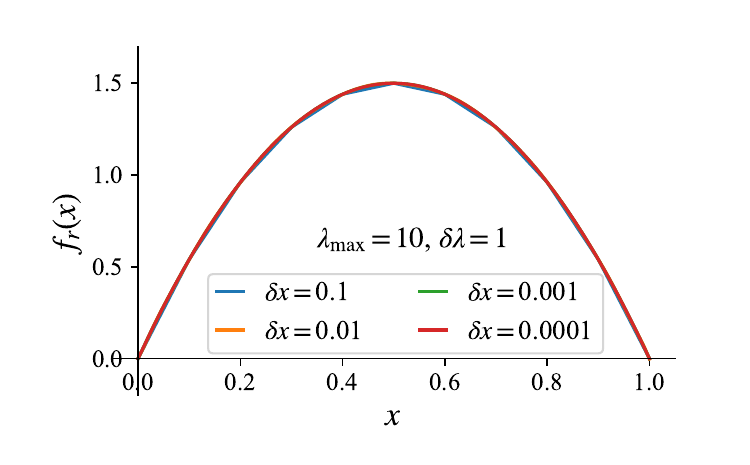}
\caption{Dependence of the resulting solutions $f_r(x)$ on $\delta\lambda$ (upper panel), $\lambda_{\mathrm{max}}$ (center panel) and $\delta x$ (lower panel). When the discretized input $g(\lambda)$ is free of errors, the results are basically independent of these parameters and closely replicate the true solution $f_t(x)$.}
 \label{fig:DFT_params}
\end{figure}

However, realistic physical inputs inevitably contain uncertainties that are significantly amplified through the discrete Fourier inversion process. This ill-posedness of DFT, mathematically characterized by significantly divergent condition numbers (e.g., $\kappa\sim10^{17}$ in the aforementioned case), implies that the unregularized solutions $f_r(x)$ is dominated by noise artifacts rather than physical signals under error-contaminated conditions. Typically, lattice QCD data contain two distinct uncertainty types: statistical uncertainties originating from finite Monte Carlo sampling (manifesting as uncorrelated noise), and systematic uncertainties arising from correlated analysis procedures such as fitting, extrapolation, and so on.

We first model the impact of uncorrelated uncertainties through controlled fluctuation of synthetic lattice data. The discretized $g^{\delta}(\lambda)$ with uncorrelated errors are constructed as:
\begin{align}
 g^\delta(\lambda_j)& = g(\lambda_j) + N(0,\sigma^2)e^{\delta M |\lambda_j|}, \nonumber\\
 &\mathrm{with~~}\lambda_j \in {-\lambda_{\text{max}}, -\lambda_{\text{max}}+\delta\lambda, \ldots, \lambda_{\text{max}}}
\end{align}
where the ground truth $g(\lambda_j)$  is generated by DFT, see Eq.~(\ref{eq:discrete_Fourier_Transform}) of the momentum-space distribution
\begin{align}
      f_t(x) = \left\{ 
           \begin{aligned}
                 6x(1-x), \quad & x\in[0,1]\\
                 0, \quad & \text{otherwise}
           \end{aligned}        
      \right.,
\end{align}
and $N(0, \sigma^2)$ represents Gaussian-distributed statistical noise with $\mu=0$ and $\sigma=0.05$. Besides, an exponentially increasing factor $e^{\delta M |\lambda_j|}$ with $\delta M=0.02$ is introduced to characterizes the signal-to-noise degradation at large Euclidean separations, analogous to hadronic matrix element measurements in lattice QCD. We generate $N_{\text{samp}} =200$ bootstrap samples of synthetic lattice data, as shown in Fig.~\ref{fig:uncorrelated-error}, to emulate the finite Monte Carlo sampling. The parameter $\delta M$ connects to hadronic spectral properties through $\delta M\sim m_H-nm_{\pi}$, where $n=1$ (for mesons) or $n=3/2$ (for baryons) governs the energy difference between hadronic states and multi-pion thresholds \cite{Wagman:2017gqi}.

\begin{figure}
\centering
\includegraphics[width=0.92\linewidth]{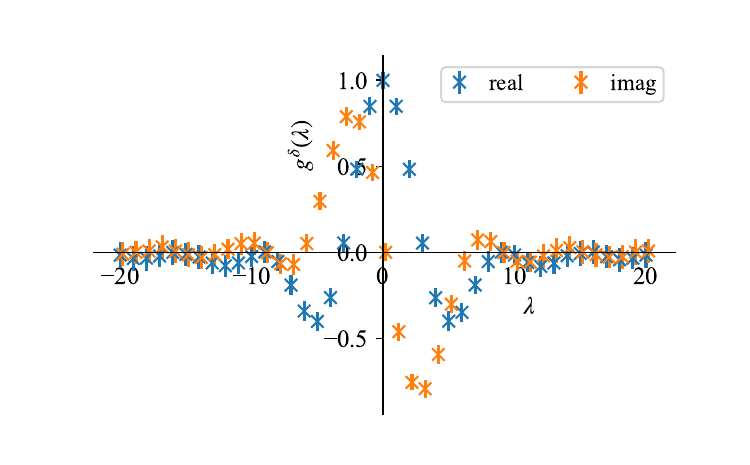}
\includegraphics[width=0.92\linewidth]{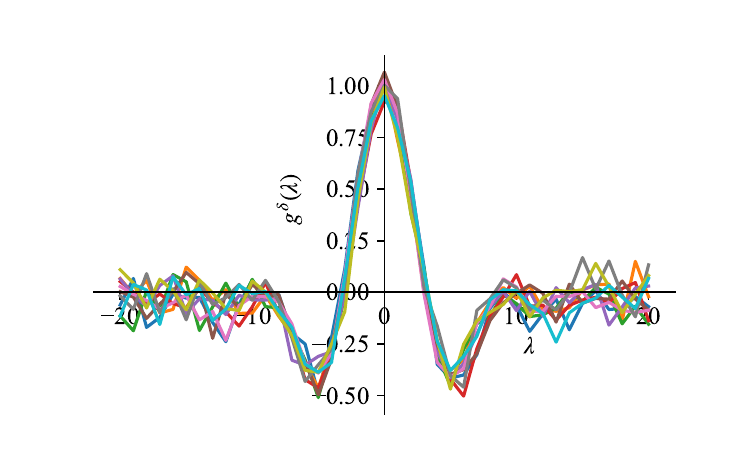}
\caption{Upper panel: A set of synthetic lattice data as bootstrap samples ($N_{\mathrm{samp}}=200$) with uncorrelated noise ($\mu = 0, \sigma = 0.05$, $ \delta M= 0.02$). Blue and orange curves represent the real and imaginary components of $g^{\delta}(\lambda)$ respectively. Lower panel: The first 10 samples exhibit that the distribution behavior under uncorrelated errors.}
\label{fig:uncorrelated-error}
\end{figure}

Unregularized DFT inversion produces divergent solutions, as illustrated by the upper panel of Fig.~\ref{fig:result-of-uncorrelated-error}, where the reconstructed distributions $f_r(x)$ exhibit significant deviations from the true solution $f_t(x)$. These unbounded oscillations originate from bootstrap ensembles ($\Delta f_r/\Delta g \sim \mathcal{O}(\kappa)$), where tiny perturbations $\delta g \sim \mathcal{N}(0,\sigma^2)$ drive uncontrolled variance amplification ($\|f_r\|_{l^2} \to \infty$). This phenomenon provides a validation of the inverse problem's violation of Hadamard's third criterion for well-posedness.

\begin{figure}
\centering
\includegraphics[width=0.92\linewidth]{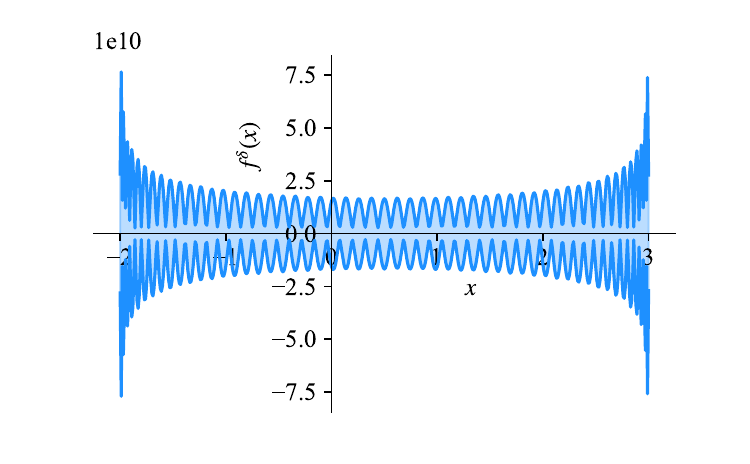}
\includegraphics[width=0.92\linewidth]{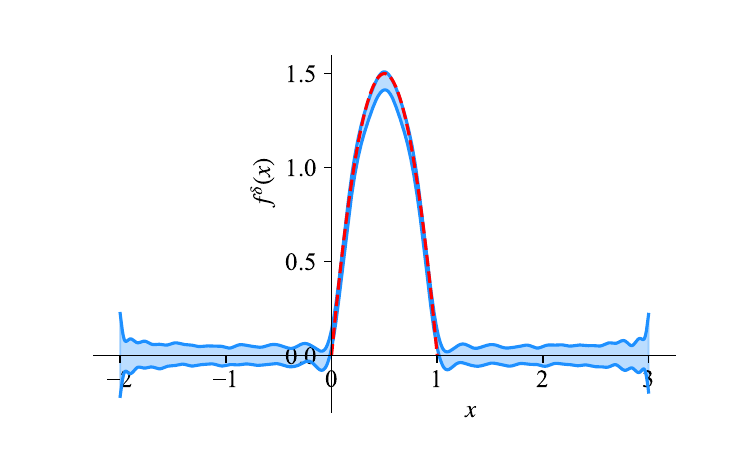}
\caption{Comparison with results from traditional SVD solution (upper panel) and Tikhonov regularization (lower panel). The red dashed lines represent the true solution $f_t(x)$.}
\label{fig:result-of-uncorrelated-error}
\end{figure}

Tikhonov regularization restores numerical stability through filtering of error-amplified modes while preserving the asymptotic $6x(1-x)$ profile, shown as the lower panel of Fig.~\ref{fig:result-of-uncorrelated-error}. The L-curve criterion optimally balances regularized solution norms $\| f^\delta_\alpha \|_{l^2}^2$ against residual norms $\|Kf_\alpha^\delta - g^\delta\|_{l^2}^2$, strongly reduce the reconstruction errors. This method effectively mitigates statistical noise propagation, demonstrating robust error control essential for extracting physical signals from noisy data. The detailed implementation and parameter selection strategies will be systematically addressed in Sec.~\ref{sec:Pion-quasi-DA}.  

We next investigate the case with correlated systematic errors described by the modified noise formula
\begin{align}
    g^\delta(\lambda_j) = g(\lambda_j) + \eta e^{\delta M |\lambda_j|}, \quad \eta \sim N(0,\sigma^2)
\end{align}
where the noise term $\eta$ maintains strict spatial coherence across all $\lambda_j$ (constant on each bootstrap sample), simulating systematic fluctuations from correlated analysis procedures. Adopting the identical parameterization ($\mu=0$, $\sigma=0.05$, $\delta M=0.02$) as for the uncorrelated scenario, we synthesize $N_{\mathrm{samp}}=200$ synthetic datasets, whose characteristic scattering is visualized in Fig.~\ref{fig:correlated-error}.

\begin{figure}
\centering
\includegraphics[width=0.92\linewidth]{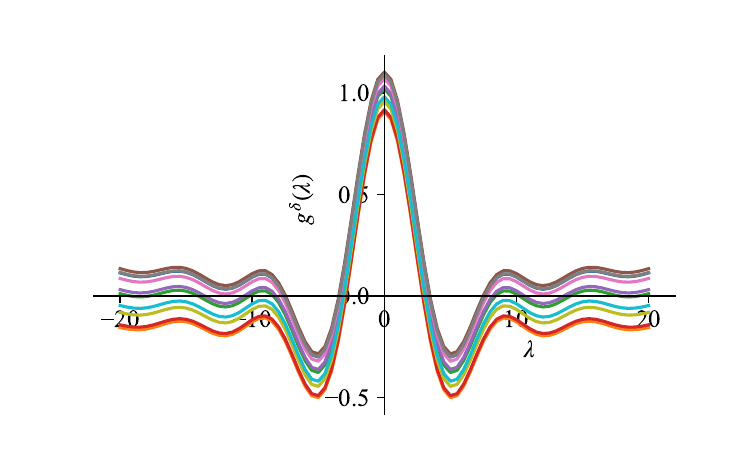}
\caption{The function $g^{\delta}$ for the first 10 samples with fully correlated errors, using the same parameter settings as in Fig.~\ref{fig:uncorrelated-error}.}
 \label{fig:correlated-error}
\end{figure}

The traditional SVD reconstruction with correlated errors, as shown in the upper panel of Fig.~\ref{fig:result-of-correlated-error}, exhibits significant suppression of high-frequency oscillations compared to the uncorrelated case, while it retains visible low-frequency distortions ($\Delta f_r/\Delta g \sim \mathcal{O}(10^3)$). This convergence behavior arises from error coherence, however, it maskes the true solution in the distorted low-frequency components. 

The Tikhonov regularization scheme suppresses the discrepancies between high- and low-frequency error components by introducing the regularization parameter $\alpha$. 
The L-curve criterion optimally selects $\alpha$ to minimize the expected reconstruction error $\mathbb{E}\left[\|f_\alpha^\delta - f_t\|_{l^2}\right]$ and achievs enhanced solution accuracy. However, correlated phases originating from the bootstrap ensemble covariance structure $\text{Cov}(\eta(\lambda_j), \eta(\lambda_k)) = \sigma^2\delta_{jk}$ induce persistent second-order oscillations that resist suppression through single-parameter regularization.

\begin{figure}
\centering
\includegraphics[width=0.92\linewidth]{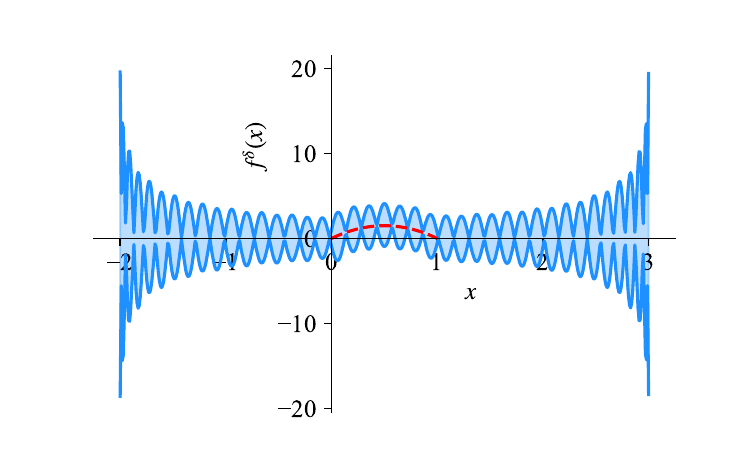}
\includegraphics[width=0.92\linewidth]{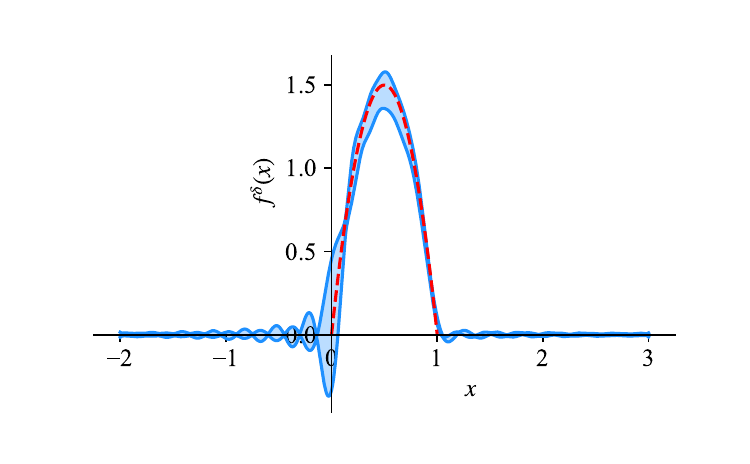}
\caption{Comparison with results from traditional SVD solution (upper panel) and Tikhonov regularization (lower panel). The red dashed lines represent the true solution $f_t(x)$.}
\label{fig:result-of-correlated-error}
\end{figure}

\section{Application to Pion Quasi Distribution Amplitude}\label{sec:Pion-quasi-DA}

\begin{figure}
\centering
\includegraphics[width=0.92\linewidth]{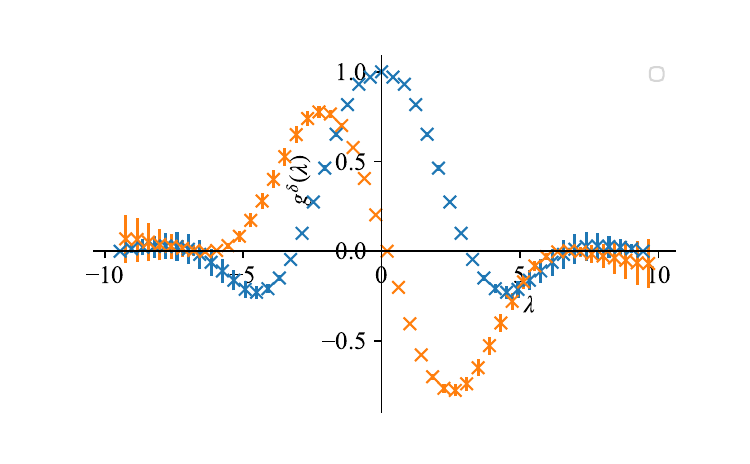}
\caption{Numerical results for the pion quasi DA from the lattice QCD calculation \cite{LatticeParton:2022zqc}. Blue and orange data points represent the real and imaginary parts respectively.}
 \label{fig:real-input-pion-quasi-DA}
\end{figure}

The toy models in Sec.~\ref{sec:Toy-model} quantitatively demonstrate the efficacy of Tikhonov regularization in mitigating both uncorrelated and correlated uncertainties within discrete Fourier inversion. Based on these tests, we implement the regularization scheme to analyze the pion quasi DA from lattice QCD, specifically utilizing the bootstrap ensembles of Euclidean correlators $g(\lambda)$ generated by the Lattice Parton Collaboration (LPC) \cite{LatticeParton:2022zqc}, as illustrated in Fig.~\ref{fig:real-input-pion-quasi-DA}.

\begin{figure}
\centering
\includegraphics[width=0.92\linewidth]{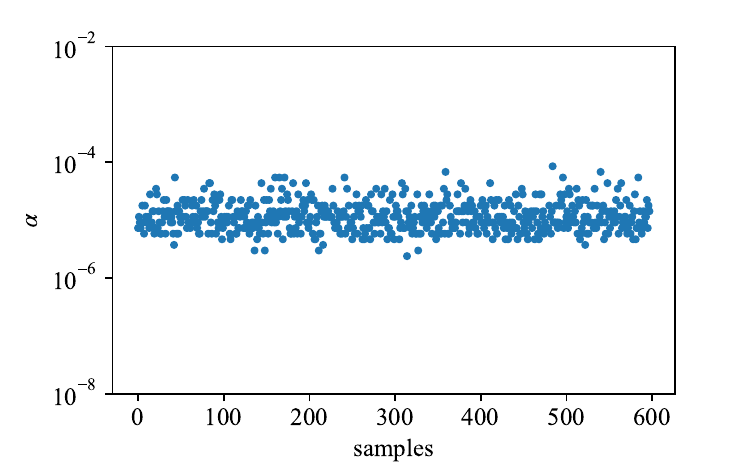}
\includegraphics[width=0.92\linewidth]{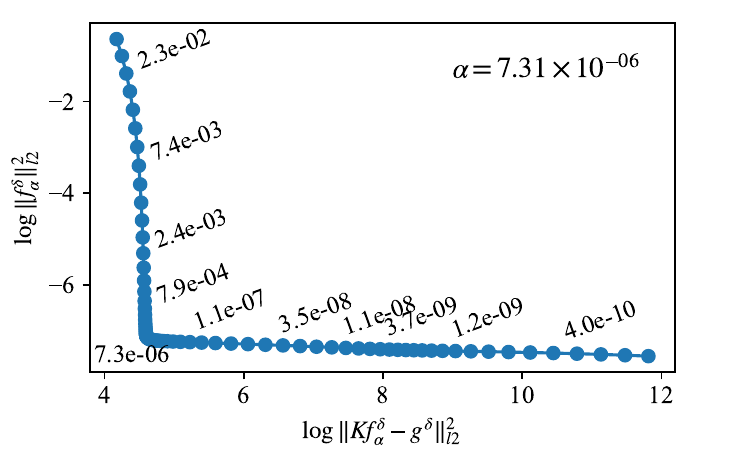}
\caption{Selection of regularization parameter $\alpha$ using the L-curve method. Upper panel: Determination of $\alpha$ on each bootstrap sample by the L-curve criterion. Lower panel: The characteristic L-curve, which represents the balance between solution norm and residual norm for one sample. The curvature maximum identifies the optimized regularization parameter $\alpha_{\text{opt}}$.}
  \label{fig:alpha-selection-for-pion-quasi-DA}
\end{figure}

\begin{figure*}
\centering
\includegraphics[width=1\linewidth]{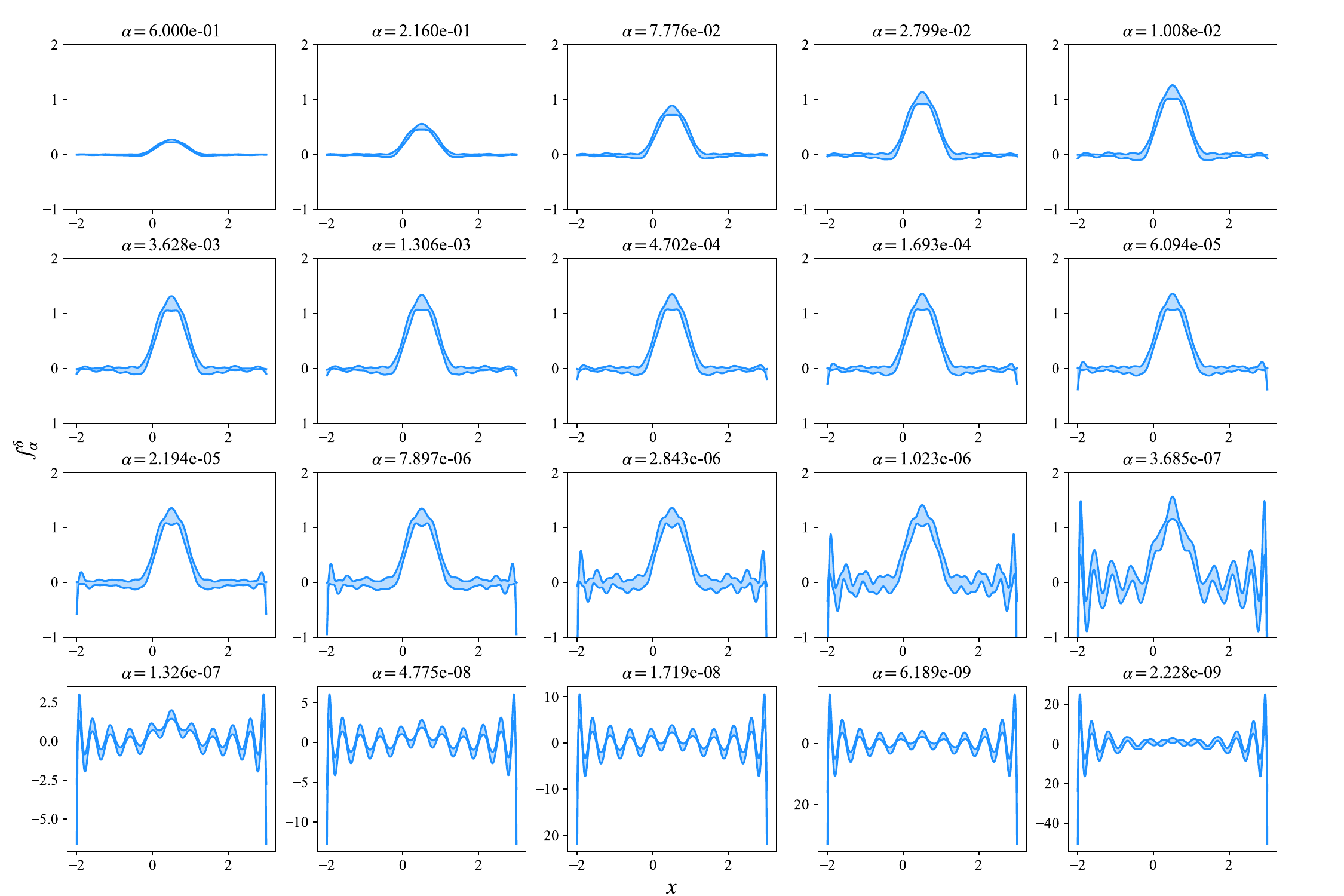}
\caption{Dependence of regularized solutions on the regularization parameters $\alpha$. The choice of $\alpha = \{0.6^{2n+1} \, | \, n = 0, 1, \ldots, 19\}$.}
  \label{fig:Tikh_alpha_evolution}
\end{figure*}

The inverse problem is solved through the following constrained optimization:
\begin{align}
    f_\alpha^\delta(x) = \arg \min_{f} \left( \|Kf - g^\delta\|_{l^2}^2 + \alpha \|f\|_{l^2}^2 \right),
\end{align}
where the regularization parameter $\alpha$ is determined through the L-curve criterion (as illustrated in Fig.~\ref{fig:alpha-selection-for-pion-quasi-DA}). The optimal value $\alpha_{\text{opt}} = 7.31 \times 10^{-6}$ corresponds to the curvature maximum in the $\log\|Kf_\alpha^\delta - g^\delta\|^2_{l^2}$ vs. $\log\|f_\alpha^\delta\|^2_{l^2}$ plane. Notably, $\alpha$ exhibits fair stability across three orders of magnitude ($10^{-6} \leq \alpha \leq 10^{-4}$), validating the numerical reliability of both the regularization framework and the L-curve selection method.

\begin{figure}
\centering
\includegraphics[width=0.92\linewidth]{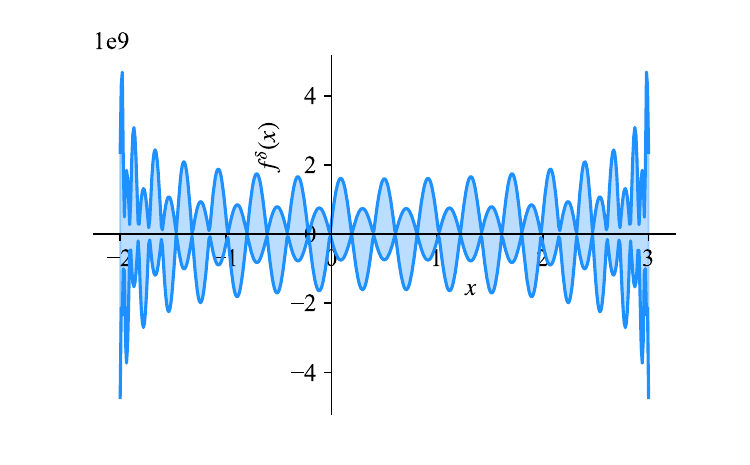}
\includegraphics[width=0.92\linewidth]{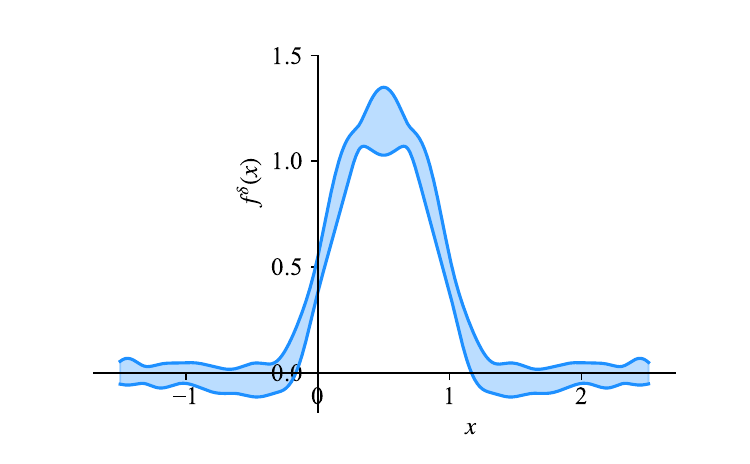}
\caption{Comparison with results of pion quasi DA from traditional SVD solution (upper panel) and Tikhonov regularization (lower panel).}
    \label{fig:results-of-pion-quasi-DA}
\end{figure}

The lower panel of Fig.~\ref{fig:alpha-selection-for-pion-quasi-DA} illustrates the L-curve, which quantifies the critical balance between the residual norm $ \log \|Kf^\delta_\alpha - g^\delta\|_{l^2}^2$ ($x$-axis) and the regularized solution norm $\log \|f^\delta_\alpha\|_{l^2}^2$ ($y$-axis). This optimized curve emerges from the competition between data fidelity (requiring close approximation to input measurements) and solution stability. The regularization parameter $\alpha$ governs this balance with three different regions, as shown in Fig.~\ref{fig:Tikh_alpha_evolution}: Oversmoothing occurs when $\alpha>\mathcal{O}(10^{-4})$, where excessive regularization leads to a systematic deviation from the true solution. Conversely, underregulation, i.e.  $\alpha<\mathcal{O}(10^{-6})$, permits uncontrolled noise amplification, manifested as high-frequency oscillations. The optimal regime $\mathcal{O}(10^{-6})<\alpha<\mathcal{O}(10^{-4})$ achieves balanced error control, which is in quantitative agreement with the L-curve optimized values listed in the upper panel of Fig.~\ref{fig:alpha-selection-for-pion-quasi-DA}.

Fig.~\ref{fig:results-of-pion-quasi-DA} compares the results of the traditional SVD solution and Tikhonov regularization. The SVD approach exhibits severe ill-posedness, manifested as unphysical oscillations with amplification factors $\Delta f_r/\Delta g \sim \mathcal{O}(10^{10})$. In  contrast, the Tikhonov regularized solutions suppress these instabilities.

\begin{figure}
\centering
\includegraphics[width=0.92\linewidth]{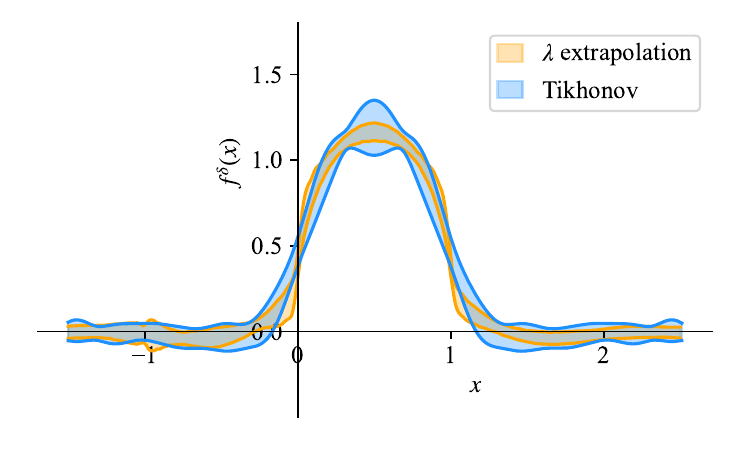}
\caption{Pion quasi DA from $\lambda$-extrapolation (orange band) and Tikhonov regularization (blue band). }
\label{fig:Comparison-with-results-from-Lamda-extra-and-Tikhonov}
\end{figure}

We also compare our results with those obtained from the $\lambda$-extrapolation method. As shown in Fig.\ref{fig:Comparison-with-results-from-Lamda-extra-and-Tikhonov}, both approaches achieve statistically consistent reconstructions within uncertainties. The $\lambda$-extrapolation method explicitly incorporates physical constraints, such as polynomial boundedness and endpoint suppression, effectively suppressing the ill-posedness when adequately precise data exist in the large $\lambda$ region.
However, in scenarios with a short limited $\lambda$ region with precise data, the Tikhonov approach provides an important complementary strategy. 
Notably, both methods converge to identical DA profiles when applied to precision data, confirming the inverse problem's well-posedness under optimal conditions. This work has demonstrated that Tikhonov regularization serving as a primary reconstruction tool for moderate-precision data sets and as a systematic cross-check for the physics-driven methods.


\section{Summary and outlook}\label{sec:Summary-outlook}
In this work, we establish a rigorous mathematical framework for addressing ill-posed inverse problems in lattice QCD calculations of partonic distributions of hadrons, focusing on the inverse limited Fourier transform as a representative case. This problem satisfies the existence and uniqueness condition and we systematically investigate the instability introduced by truncation. We demonstrated that this instability can be effectively addressed by Tikhonov regularization without relying on ansatz-based assumptions.
Through theoretical discussions and numerical validations for toy models with both uncorrelated and correlated noise, we confirm that Tikhonov regularization together with L-curve parameter selection provides a robust mechanism for stabilizing solutions while preserving essential physical features. The successful application to pion quasi DA lattice data further validates the method, with results consistent with the physical driven $\lambda$-extrapolation approach.  

Our comparative analysis reveals distinct operational domains for different inversion strategies, validating that both physics-driven and mathematics-driven approaches can achieve reliable results for sufficient data precision. 
{However, both methods come with their respective advantages and limitations. 
For the physically physics-driven $\lambda$-extrapolation approach, it requires reliable signal quality in the large $\lambda$ region ($\lambda > 10$). When applicable, this method can completely circumvent the inverse problem and provide reliable uncertainty estimates, as pointed out in Ref.~\cite{Chen:2025cxr}. However, it is well known that the uncertainties in lattice data do not grow linearly, and performing extrapolation in regions with insufficiently large $\lambda$ may introduce additional systematic errors. 
On the other hand, the Tikhonov regularization method can still perform well if the signal sufficiently accurate in the smaller $\lambda$ region $\lambda < 10$. Nonetheless, the introduction of a regularization parameter may lead to potential systematic uncertainties, which must be carefully handled.}

{For the next stage, in addition to improving the accuracy of lattice calculation, two frontier challenges are prominent. First, the combination of physics-driven priors (e.g., endpoint behavior constraints from operator product expansion) with inverse method such as Tikhonov regularization may further balance the stability and accuracy of solution.}
Advances in inverse problem theory, particularly in uncertainty quantification for nonlocal operators, could further bridge the gap between lattice calculated correlators and partonic distributions of hadrons. These developments, coupled with anticipated improvements in lattice data precision, position inverse problem techniques as increasingly central to precision QCD studies in the exascale computing era.

\section{Acknowledgments}

We thank Andreas Sch\"afer and Xiangdong Ji for valuable discussions and constructive suggestions on this work.
We also thank Xiong-Bin Yan for valuable  discussions in mathematics.   We thank the CLQCD collaborations for providing us the lattice data of pion distribution amplitudes \cite{LatticeParton:2022zqc}. This work is supported in part by Natural Science Foundation of China under grant No. 12375069 and 12335003.  Q.A.Z is supported by the Fundamental Research Funds for the Central Universities. J.H is supported by Guang-dong Major Project of Basic and Applied Basic Research No. 2020B0301030008. A.S.X and F.S.Y are supported by the Fundamental Research Funds for the Central Universities under No.~lzujbky-2023-stlt01, lzujbky-2024-oy02 and lzujbky-2025-eyt01.

\appendix 
\begin{widetext}
\section{The Proof of the uniqueness}\label{sec:Proof of the uniqueness}
\begin{theorem}[The Uniqueness of the Limited Fourier Transform] 
    Suppose that $f_1(x),f_2(x) \in L^2(x_{\mathrm{min}},x_{\mathrm{max}})$, with $-\infty<x_{\mathrm{min}}<x_{\mathrm{max}}<\infty$. If $\int_{x_{\mathrm{min}}}^{x_{\mathrm{max}}} e^{ix\lambda} f_1(x) dx=\int_{x_{\mathrm{min}}}^{x_{\mathrm{max}}} e^{ix\lambda} f_2(x) dx=g(\lambda),\lambda\in [\lambda_{\mathrm{min}},\lambda_{\mathrm{max}}]$ (or a sequence of points converging to this interval), with $-\infty<\lambda_{\mathrm{min}}<\lambda_{\mathrm {max}}<\infty$, then we have $f_{1}(x)=f_{2}(x)$, a. e. $x\in [x_{\mathrm{min}},x_{\mathrm{max}}]$.
\label{thm:The Proof of the uniqueness}
\end{theorem}

\begin{proof} 
Since the integral equation is linear, we know that 
\begin{align}
    \int_{x_{\mathrm{min}}}^{x_{\mathrm{max}}} e^{ix\lambda} \left(f_1(x)-f_2(x)\right) dx=0.
\end{align} 
Setting $f(x)=f_1(x)-f_2(x)$, we just need to prove that 
\begin{align}\label{eqq}
   \int_{x_{\mathrm{min}}}^{x_{\mathrm{max}}} e^{ix\lambda} f(x) dx=0, ~\lambda\in [\lambda_{\mathrm{min}},\lambda_{\mathrm{max}}] 
\end{align}
implies $f(x)=0$, a. e. $x\in [x_{\mathrm{min}},x_{\mathrm{max}}]$.

First, we prove the analyticity of the complex function 
\begin{align}
    \phi(z)=\int_{x_{\mathrm{min}}}^{x_{\mathrm{max}}} e^{ix z} f(x) dx,z\in \mathbb{C}.
\end{align}

It is known that $e^{ixz}$ is analytic for $z\in \mathbb C$ to $x\in \mathbb R$, and $e^{ixz}=\sum_{n=0}^\infty \frac{i^n}{n!}z^n x^n$. Then we have 
\begin{align}\label{ser}
     \phi(z)= \int_{x_{\mathrm{min}}}^{x_{\mathrm{max}}} e^{ixz}  f(x) dx=\sum_{n=0}^\infty \frac{i^n}{n!}z^n \int_{x_{\mathrm{min}}}^{x_{\mathrm{max}}} x^n f(x) dx, 
\end{align}
For $f\in L^1(x_{\mathrm{min}},x_{\mathrm{max}})$, we have 
\begin{align}
    \left|\int_{x_{\mathrm{min}}}^{x_{\mathrm{max}}} x^n f(x) dx\right| \le \left(\max\{|x_\mathrm{min}|,|x_\mathrm{max}|\}\right)^n \int_{x_{\mathrm{min}}}^{x_{\mathrm{max}}}|f(x)|dx.
\end{align}
The coefficients in the series $\phi(z)$ satisfy
\begin{align}
    \lim_{n\rightarrow \infty} \left(\frac{1}{n!}\right)^{1/n} \left|\int_{x_{\mathrm{min}}}^{x_{\mathrm{max}}} x^n f(x) dx\right|^{1/n}\le 
    \lim_{n\rightarrow \infty} \left(\frac{1}{n!}\right)^{1/n} \left(\int_{x_{\mathrm{min}}}^{x_{\mathrm{max}}}|f(x)|dx\right)^{1/n}(\max\{|x_\mathrm{min}|,|x_\mathrm{max}|\})
\rightarrow 0,
\end{align}
by $\lim_{n\rightarrow\infty}  ( n!)^{1/n}=\infty$. Therefore 
the convergence radius of the series $\phi(z)$ is $\infty$, that means 
the complex function $\phi(z)$ is analytic in the complex plane. By the property of analytic function, from Eq.~(\ref{eqq}), we have 
 $\phi(z)=0, ~z\in \mathbb{C}$, namely
    \begin{align}
\int_{x_{\mathrm{min}}}^{x_{\mathrm{max}}}f(x) dx + i z \int_{x_{\mathrm{min}}}^{x_{\mathrm{max}}} x f(x) dx + \frac{z^2}{2} i^2 \int_{x_{\mathrm{min}}}^{x_{\mathrm{max}}} x^2 f(x) dx  + \cdots + \frac{z^n}{n!} i^n\int_{x_{\mathrm{min}}}^{x_{\mathrm{max}}} x^n f(x) dx + \cdots =0, ~z\in \mathbb{C}. \label{eq:Taylor and unique I}
    \end{align}
    
From Eq.~(\ref{eq:Taylor and unique I}) we obtain for $z=0$
\begin{align}
\int_{x_{\mathrm{min}}}^{x_{\mathrm{max}}}f(x) dx=0    
\end{align}
and 
\begin{align}\label{eq:Taylor and unique II}
        & iz \int_{x_{\mathrm{min}}}^{x_{\mathrm{max}}} x f(x) dx + i^2\frac{z^2}{2} \int_{x_{\mathrm{min}}}^{x_{\mathrm{max}}}  x^2 f(x) dx  + \cdots + i^n \frac{z^n}{n!} \int_{x_{\mathrm{min}}}^{x_{\mathrm{max}}} x^n f(x) dx + \cdots =0, ~z\in \mathbb{C}.
\end{align}

Dividing both sides of Eq.~(\ref{eq:Taylor and unique II}) by $iz$, and then set $z=0$, one obtain 
\begin{align}
    \int_{x_{\mathrm{min}}}^{x_{\mathrm{max}}}xf(x) dx=0.
\end{align} 
Repeating above process, one can obtain that
\begin{align}\label{The moment problem}
    \int_{x_{\mathrm{min}}}^{x_{\mathrm{max}}}x^n f(x) dx=0, ~~n=0,1,2,\cdots
\end{align}

Since $C[x_{\mathrm{min}},x_{\mathrm{max}}]$ is dense in $L^{2}(x_{\mathrm{min}},x_{\mathrm{max}})$, then for $f(x)\in L^{2}(x_{\mathrm{min}},x_{\mathrm{max}})$ and any $\epsilon>0$, there exists $\tilde{f}(x) \in C[x_{\mathrm{min}},x_{\mathrm{max}}]$, such that 
\begin{align}
||f-\tilde{f}||_{L^{2}(x_{\mathrm{min}},x_{\mathrm{max}})}<\epsilon.    
\end{align} 
On the other hand, for $\tilde{f}(x) \in C[x_{\mathrm{min}},x_{\mathrm{max}}]$, there exists a polynomial $Q_{n}(x)$ of degree $n\in \mathbb{N}$, such that
\begin{align}
||\tilde{f}-Q_{n}||_{C[x_{\mathrm{min}},x_{\mathrm{max}}]}<\epsilon,
\end{align}
by the Weierstrass theorem. Therefore, one have
\begin{align}\label{Weierstrass}
\left\|f-Q_n\right\|_{L^2(x_{\mathrm{min}},x_{\mathrm{max}})} & \leq\|f-\tilde{f}\|_{L^2(x_{\mathrm{min}},x_{\mathrm{max}})}+\left\|\tilde{f}-Q_n\right\|_{L^2(x_{\mathrm{min}},x_{\mathrm{max}})}
\nonumber\\ 
& \leq \epsilon+\sqrt{x_{\mathrm{max}}-x_{\mathrm{min}}}\left\|\tilde{f}-Q_n\right\|_{C[x_{\mathrm{min}},x_{\mathrm{max}}]} \nonumber\\
& <\epsilon+\epsilon \sqrt{x_{\mathrm{max}}-x_{\mathrm{min}}}.
\end{align}

By using Eq.~(\ref{The moment problem}), we know that 
\begin{align}
\int_{x_{\mathrm{min}}}^{x_{\mathrm{max}}}f(x)Q_{n}(x)dx=0.    
\end{align} 
Combined with the Cauchy inequality, we have
\begin{align}
\|f\|_{L^2(x_{\mathrm{min}},x_{\mathrm{max}})}^2 & =\int_{x_{\mathrm{min}}}^{x_{\mathrm{max}}} f^2(x) d x=\int_{x_{\mathrm{min}}}^{x_{\mathrm{max}}} \left(f^2(x)-f(x) Q_n(x)\right) d x \nonumber\\
& \leq \int_{x_{\mathrm{min}}}^{x_{\mathrm{max}}} |f(x)| \cdot\left|f(x)-Q_n(x)\right| d x \nonumber\\
& \leq\left(\int_{x_{\mathrm{min}}}^{x_{\mathrm{max}}} f^2(x) d x\right)^{\frac{1}{2}}\left(\int_{x_{\mathrm{min}}}^{x_{\mathrm{max}}} \left|f(x)-Q_n(x)\right|^2 d x\right)^{\frac{1}{2}} \nonumber\\
& =\|f\|_{L^2(x_{\mathrm{min}},x_{\mathrm{max}})}\left\|f-Q_n\right\|_{L^2(x_{\mathrm{min}},x_{\mathrm{max}})} \nonumber\\
& \leq(\epsilon+\epsilon \sqrt{x_{\mathrm{max}}-x_{\mathrm{min}}})\|f\|_{L^2(x_{\mathrm{min}},x_{\mathrm{max}})},
\end{align}
which implies that 
\begin{align}
    \|f\|_{L^{2}(x_{\mathrm{min}},x_{\mathrm{max}})}\leq \epsilon+\epsilon\sqrt{x_{\mathrm{max}}-x_{\mathrm{min}}}.
\end{align}

Letting $\epsilon\rightarrow 0$, we finally have
\begin{align}
    \|f\|_{L^{2}(x_{\mathrm{min}},x_{\mathrm{max}})}=0,
\end{align}
i. e. $f(x)=0$, a. e. $x\in [x_{\mathrm{min}},x_{\mathrm{max}}]$. The proof is completed.

\end{proof}
\end{widetext}


\end{document}